\documentclass[reqno,a4paper,11pt]{amsart}
\usepackage{times} 
\usepackage{algorithm}
\usepackage{algpseudocode}
\usepackage{amsmath,amssymb,amsfonts,mathrsfs}
\usepackage{a4wide}
\usepackage{booktabs}
\usepackage{hyperref}
\usepackage[english]{babel} 
\usepackage{parskip}
\usepackage{amsthm}

\theoremstyle{plain}
 
\newtheorem{theorem}{Theorem}
\newtheorem{proposition}[theorem]{Proposition}

\newtheorem{lemma}[theorem]{Lemma}

\theoremstyle{definition}
\newtheorem{definition}[theorem]{Definition}

\theoremstyle{remark}
\newtheorem{remark}[theorem]{Remark}

\renewcommand{\leq}{\leqslant}
\renewcommand{\le}{\leqslant}
\renewcommand{\geq}{\geqslant}
\renewcommand{\ge}{\geqslant}

\newcommand{\A}{\mathbb{A}}
\newcommand{\F}{\mathbb{F}}
\newcommand{\K}{\mathbb{K}}
\renewcommand{\L}{\mathbb{L}}
\newcommand{\fq}{\F_{q}}
\newcommand{\fqm}{\F_{q^m}}

\newcommand{\word}[1]{\mathbf{#1}}
\newcommand{\av}{\word{a}}
\newcommand{\bv}{\word{b}}
\newcommand{\cv}{\word{c}}
\newcommand{\ev}{\word{e}}
\newcommand{\gv}{\word{g}}
\newcommand{\hv}{\word{h}}
\newcommand{\mv}{\word{m}}
\newcommand{\sv}{\word{s}}
\newcommand{\uv}{\word{u}}
\newcommand{\vv}{\word{v}}
\newcommand{\xv}{\word{x}}

\newcommand{\zv}{\word{z}}
\newcommand{\zz}{\word{0}}
\newcommand{\Kv}{\word{K}}

\newcommand{\mat}[1]{\mathbf{#1}}
\newcommand{\Am}{\mat{A}}
\newcommand{\Bm}{\mat{B}}
\newcommand{\Cm}{\mat{C}}
\newcommand{\Dm}{\mat{D}}
\newcommand{\Em}{\mat{E}}
\newcommand{\Gm}{\mat{G}}

\newcommand{\Km}{\mat{K}}
\newcommand{\Lm}{\mat{L}}

\newcommand{\Pm}{\mat{P}}

\newcommand{\Rm}{\mat{R}}

\newcommand{\Tm}{\mat{T}}
\newcommand{\Vm}{\mat{V}}

\newcommand{\ZZ}{\mat{0}}
\newcommand{\Gp}{\mat{G}_{\rm pub}}
\newcommand{\Cp}{\code{C}_{\rm pub}}
\newcommand{\tp}{t_{\rm pub}}

\newcommand{\gab}[2]{\CG_{#1}\left(#2\right)}
\newcommand{\code}[1]{\mathscr{#1}}
\newcommand{\dual}[1]{{#1}^\bot}
\newcommand{\pdual}[1]{\left(#1\right)^\bot}

\newcommand{\CA}{\code{A}}
\newcommand{\CB}{\code{B}}
\newcommand{\CC}{\code{C}}
\newcommand{\CG}{\code{G}}
\newcommand{\Cpub}{\CC_{\text{\rm pub}}}

\newcommand{\norm}[2]{\left | {#2}  \right |_{#1}}

\newcommand{\inpro}[2]{\left \langle #1 , #2  \right\rangle} 

\newcommand{\rank}{{\normalfont \texttt{rank}}}
\newcommand{\GL}{{\normalfont \textsf{GL}}}
\newcommand{\MS}[3]{\mathcal{M}_{#1,#2}\left(#3\right)}
  
\newcommand{\Trqm}{\mathbf{Tr}_{\fqm/ \fq}}
\newcommand{\Tr}{\mathbf{Tr}_{\L/\fqm}}

\title{Polynomial-Time Key Recovery Attack on the Faure-Loidreau Scheme based on Gabidulin Codes}

\author{Philippe Gaborit}

	\address{ Philippe Gaborit  is with
	XLIM-DMI, Universit\'e de Limoges, 123, 
		Avenue Albert Thomas, F-87060, Limoges Cedex, France.}
	\email{gaborit@unilim.fr}
	
\author{Ayoub Otmani}
	\address{Ayoub Otmani and Herv\'e Tal\'e Kalachi  are with the University of Rouen, 
 		UFR des Sciences et des Techniques,
  		BP 12, Avenue de l'Universit\'e,
  		F-76801 Saint-\'Etienne-du-Rouvray Cedex, France.}
	\email{ayoub.otmani@univ-rouen.fr}

\author{Herv\'e Tal\'e Kalachi} 
 	\address{Herv\'e Tal\'e Kalachi  is with
	the University of Yaounde 1, Department of Mathematics, ERAL, Cameroon.}
	\email{hervekalachi@gmail.com}

\begin{document}
\begin{abstract}
	Encryption schemes based on the rank metric lead to small public key sizes of order of few thousands bytes which represents a very 
attractive feature compared to Hamming metric-based encryption schemes where public key sizes are of order of hundreds of thousands 
bytes even with additional structures like the cyclicity.
The main tool for building public key encryption schemes in rank metric is the McEliece encryption setting used with the  family of Gabidulin codes. 
Since the original scheme proposed in 1991 by Gabidulin, Paramonov and Tretjakov, many systems have been proposed 
based on different masking techniques 
for Gabidulin codes. 
Nevertheless, over the years 
most of these systems were attacked essentially by the use of an attack proposed by Overbeck.

In 2005 Faure and Loidreau designed a rank-metric encryption scheme which was not in the McEliece setting. The scheme is very efficient, 
with small public keys of size a few kiloBytes and with security closely related to  the linearized polynomial reconstruction problem 
which corresponds to the decoding problem of Gabidulin codes. The structure of the scheme differs considerably from the classical McEliece 
setting and until our work, the scheme had never been attacked. We show in this article that 
for a range of parameters,
this scheme
is also  vulnerable to a polynomial-time attack that recovers the private key by applying Overbeck's attack on an 
appropriate public code. As an example we break in a few seconds parameters with $80$-bit security claim.
Our work also shows that some parameters are not affected by our attack but at the cost of a lost of efficiency for the underlying schemes.
\end{abstract}

\maketitle

\keywords{Post-quantum cryptography; Gabidulin code; Polynomial reconstruction; Faure-Loidreau scheme.}


\section{Introduction}

\paragraph{\bf McEliece encryption setting.} Post-quantum cryptography aims at proposing schemes that resist to an hypothetical quantum computer.
It represents more and more a serious alternative to classical cryptography based on 
the discrete logarithm problem and the factorization problem. McEliece 
opened the way to code-based cryptography by proposing the first post-quantum 
(public-key encryption) scheme \cite{M78}. The McEliece cryptosystem is in fact
an encryption setting which relies on the hiding of particular class of decodable
codes.  The algorithmic assumption underlying the security 
is the difficulty of solving the closest vector problem with the Hamming metric for the particular class of masked decodable codes on which the scheme relies. 
Over the years many variants of the McEliece cryptosystem were proposed
with different families of codes, and many were broken by recovering the structure of the masked codes. However the original family of codes, the binary Goppa codes, proposed by McEliece essentially remains unattacked. The resistance to structural attacks, which try to recover the structure of the masked codes,  is the main potential weakness of this setting. For instance the highly structured Reed-Solomon codes are difficult to mask and most of McEliece variants relying on Reed-Solomon codes or variations on Reed-Solomon codes have been broken.

\paragraph{\bf Rank metric cryptography.} The McEliece cryptosystem setting is very versatile and only needs a decodable family of codes along with a particular masking technique of codes. 
Hence this approach can also be used with another metric than the classical Hamming metric.
An important metric emerging in cryptography is the rank metric which considers the ambient space $\F^{a b}$ where $\F$ is a (finite) field and $a$ and $b$ are positive integers, as the space of $a \times b$ matrices so that we can associate the rank to any vector from $\F^{ab}$. 
By viewing any finite extension of finite fields $\F / \K$ 
as a linear space over $\K$ of dimension $m > 1$ then for any positive integer $n$, the ambient space $\F^n$ can also be viewed as the space of $m \times n$ matrices. 
In \cite{GPT91} Gabidulin, Paramonov and Tretjakov proposed
the first rank-metric based encryption scheme. 
This scheme  can be seen as an analog of the McEliece's one but  based on the class of Gabidulin codes. 

The main interest of the rank metric is that the 
time complexity of best known generic attacks for rank metric grows faster regarding 
the size of parameters, than for Hamming metric. In practice, without additional
structure like cyclicity, it means that it is possible to obtain public key sizes for rank metric 
of only a few thousand bytes, when hundred of thousand bytes are needed for Hamming metric.

An important operation in the key generation of the GPT cryptosystem is the masking phase where 
the secret Gabidulin code  $\CG$ undergoes a  transformation to mask its inherent algebraic structure.
This transformation  is  a  probabilistic algorithm  that adds some randomness to its input $\CG$.
Originally, the authors  in \cite{GPT91} proposed to use  a \emph{distortion} 
transformation  that outputs (a generator matrix of) the code   
$\CG + \code{R}$ where $\code{R}$ is random code with a prescribed dimension 
$t_R$. The presence of $\code{R}$ has however an impact: 
the sender has to  add an error vector whose rank weight is $\tp = t - t_R$ where $t$ is the error correction capability of the $\CG$. Hence, roughly speaking,  the hiding phase publishes a degraded code in terms of 
error correction.

Gabidulin codes are often seen as  equivalent of Reed-Solomon codes because, like them, 
they are highly structured. That is the reason why their use in the GPT cryptosystem has been the subject to
several attacks. Gibson was the first to prove
the weakness of the system through a series of successful attacks \cite{G95,G96}.
Following these failures, the first works which modified the GPT scheme to avoid Gibson's attack were published in \cite{GO01,GOHA03}.  
The idea is to hide further the structure of Gabidulin code by considering  isometries for the rank metric.
Consequently, a \emph{right column scrambler} $\Pm$ is  introduced which is an invertible matrix with its entries in the base field $\fq$
while the ambient space of the Gabidlun code is $\fqm^n$. 
But Overbeck designed in \cite{O05,O05a,O08}  a more
general attack that dismantled all the existing  modified GPT cryptosystems. 
His approach consists in applying an operator $\Lambda_i$ which applies $i$ times  the Frobenius operation 
on the public generator matrix $\Gp$. The dimension increases by $1$ each time the Frobenius is applied. Therefore
by taking $i = n - k - 1$ the codimension becomes $1$ if $k$ is the rank of $\Gp$.
This phenomenon is a clearly distinguishing property of a Gabidulin code which cannot be encountered for instance with a random linear code where the dimension would increase by $k$ for each use of the Frobenius operator.

Overbeck's attack uses crucially  two important facts, namely the column scrambler matrix $\Pm$ is defined on the 
based field $\fq$ and the codimension of $\Lambda _{n-k-1}\left(\Gp \right)$ is equal to $1$.
Several works then proposed to resist to this attack  either by taking special random codes $\code{R}$
so that the second property is not true as in \cite{L10,RGH10}, or 
by taking a column scrambler matrix defined over the extension field $\fqm$ as in \cite{G08,GRH09,RGH11}.

But recently in \cite{OTN16} it  was  shown that even if the column scrambler is defined 
on the extension field as in  \cite{G08,GRH09,RGH11}, by using precisely Overbeck's technique,
it is still possible to recover very efficiently a secret Gabidulin code  whose error correction $t^*$ is certainly strictly less than the error correction of the secret original Gabidiulin code
but still strictly greater than the number of added errors  $\tp$. In other words, an attacker is still able to decrypt any ciphertext 
and consequently, all schemes based on Gabidulin codes presented in \cite{G08,GRH09,RGH11} are 
actually not secure.

\paragraph{\bf Faure-Loidreau's approach.} Besides the McEliece setting used with Gabidulin codes, Faure and Loidreau proposed in \cite{FL05} another approach for designing rank-metric encryption scheme based on Gabidulin codes. The scheme was supposed to be secure under the assumption that the problem of the \emph{linearized polynomial reconstruction}\footnote{In \cite{FL05} the problem is termed as $p$-polynomial reconstruction problem.} is intractable. 
This scheme follows the works done in  \cite{AF03,AFL03} where a public-key encryption scheme is defined 
that relies on the \emph{polynomial reconstruction} problem which corresponds to the decoding problem of
Reed-Solomon codes. The Polynomial Reconstruction (PR) consists in solving the following problem: \emph{given two $n$-tuples 
$(z_1,\dots{},z_n)$ and  $(y_1,\dots{},y_n)$ and parameters $[n, k, w]$, 
recover all polynomials $P$ of degree less than $k$ such that $P(z_i) = y_i$
for at most $w$ distinct indices $i \in \{1, \dots{} , n \}$}.
The public key is then a noisy random codeword from a Reed-Solomon code where the (Hamming) weight of the error is 
greater than the decoding capability of the Reed-Solomon code.
However the schemes have undergone polynomial-time attacks in \cite{C03,C04,KY04}. 
The authors  in \cite{FL05} proposed  an analog of Augot-Finiasz scheme but  in the rank-metric context. The security of \cite{FL05}
is related to the difficulty of solving $p$-polynomial reconstruction corresponding actually
 to the decoding problem of 
a Gabidulin code 
beyond its error-correcting capability. 
After Overbeck's attack, parameters proposed in \cite{FL05} were updated in \cite[Chap. 7]{L07} in order to resist to it. 

\paragraph{\bf Our results.}
We show in this article that the Faure-Loidreau scheme is vulnerable to a structural polynomial-time attack that recovers 
the private key from the public key. Based in part on the security analysis given in \cite[Chap. 7]{L07}, 
we show that by applying Overbeck's attack on an appropriate public code an attacker can recover the private key very efficiently, only assuming a mild condition on the code, which was always true in all our experimentations. 

Informally, the Faure-Loidreau encryption scheme considers three 
finite fields $\fq \subset \fqm \subset \L$.  The rank weight of vectors is computed over the field $\fq$.
The public key is  then composed of a Gabidulin code 
of dimension $k$  of length $n$ defined by a matrix $\Gm = ( g_{i,j})$ with $g_{i,j}  \in \fqm^n$ and 
$\Kv = \xv \Gm + \zv$ where $\xv$ is some vector  in $\L^k$ and $\zv$ is a vector of $\L^n$ with (rank) weight $w >\frac{1}{2}(n-k)$. 
Both vectors $\xv$ and $\zv$ have to be kept secret but from attacker's point of view the private key is \emph{essentially} $\xv$ since $\zv$ can be deduced from it.

Our attack uses the Frobenius operator, introduced by Overbeck, which takes as input any vector space 
 $U \subseteq \fqm^n$ and integer $i \ge 1$ in order to construct the vector space $\Lambda_i(U)$ defined as
\[
\Lambda_i(U) = U + U^q + \cdots{} + U^{q^i}.
\]
The first step of the attack  considers a basis $\gamma_1,\dots{}, \gamma_u$ of $\L$ viewed as a vector space over $\fqm$ of dimension $u > 1$ 
and  defines the vectors $\vv_i  = \Tr(\gamma_i \zv)$.  Our main result shows that the system can be broken in polynomial time and can be stated as follows:

\begin{theorem} \label{thm:attack}
If the $\fqm$-vector  space generated by $\vv_1,\dots{},\vv_u$ denoted by $V$ satisfies the property 
\begin{equation} \label{eq:dimOnB}
 \dim \Lambda_{n - w - k - 1}(V) = w
\end{equation}
then the private key $(\xv,\zv)$ can be recovered from $(\Gm,\Kv)$ with $O(n^3)$ operations  in the field $\L$.
 \end{theorem}
Notice that if $V$ behaves as random code then generally the condition \eqref{eq:dimOnB} holds.
We implemented our attack on parameters given in \cite{FL05,L07} for $80$-bit security, which were broken in a few seconds. 
A necessary condition for \eqref{eq:dimOnB} to be true is to choose 
$u(n - w - k ) \geq w$ that is to say
 \[
 w \leq \frac{u}{u+1} \left( n - k \right).
 \] 
This  was always the case for parameters proposed in \cite{FL05,L07}.

\paragraph{\bf Related work.}
The attack presented in this paper is very similar to the approach proposed 
in \cite{LO06} where the authors seek to decode several noisy codewords of a Gabidulin code.
Let us assume that we received $\ell$ words $\zv_1,\dots{},\zv_\ell$  from $\fqm^n$ where 
each $\zv_i$ is written as $z_i = \cv_i + \ev_i$  with $\cv_i$ belonging to a Gabidulin code $\CG$ of dimension $k$ and length $n$
over $\fqm$ and the $\ev_i$'s are vectors from $\fqm^n$. Let us denote by $\Em$ the matrix of size $\ell \times n$ formed by the $\ev_i$'s
and let $\norm{q}{\Em}$ be the dimension of the $\fq$-vector space generated by the columns of $\Em$.
The authors show that when $\norm{q}{\Em} \leq  \frac{\ell}{\ell+1} \left( n - k \right)$
then Overbeck's technique recovers in $O(n^3)$ operations the codewords $\cv_1,\dots{},\cv_\ell$. It therefore provides a method that decodes 
a Gabidulin code beyond the classical error-correcting limit $\frac{1}{2} \left( n - k \right)$.
This approach can be used here to attack the Faure-Loidreau scheme  \cite{FL05} because 
the vectors $\Tr(\gamma_1 K),\dots{},\Tr(\gamma_u K)$ can be written as
$\cv_1+\vv_1,\dots{},\cv_u+\vv_u$  where $\cv_i = \Tr(\gamma_i \xv) \Gm$ belong to the Gabidulin generated by $\Gm$ 
and the $u \times n$ matrix $\Vm$ formed by $\vv_1,\dots{},\vv_u$ satisfy 
$\norm{q}{\Vm} = w$ which in turn has to verify $w \leq  \frac{u}{u+1} \left( n - k \right)$.

\paragraph{\bf Organisation.} In Section~\ref{sec:prelim} notation and important notions useful for our paper are given. Gabidulin codes are recalled in \ref{sec:gab}. 
In Section~\ref{sec:scheme} we present the Faure-Loidreau scheme and 
in Section~\ref{sec:attack} we describe in full details our attack against it.

\section{Preliminaries} \label{sec:prelim}

Vectors from $\F^n$ where $\F$ is a field are denoted by boldface letters as $\av = (a_1,\dots{},a_n)$. The concatenation
of two vectors $\uv$ and $\vv$ is denoted by $\left( \uv \mid \vv\right)$.
The set of matrices with entries in $\F$ having $m$ rows and $n$ columns is denoted by $\MS{m}{n}{\F}$ and 
the subset of $n \times n $ invertible matrices  form the general linear group denoted by $\GL_n(\F)$. 
A \emph{linear code $\CC$ of length} $n$ over a field $\F$ is a linear subspace of $\F^n$. 
An element of a code is called a \emph{codeword} and a matrix whose rows form a basis is 
called a \emph{generator matrix}.
The \emph{dual} of a code $\CC \subset \F^n$ is the linear space denoted by $\dual{\CC}$ containing 
vectors $\zv \in \F^n$ such that:
\[
\forall \cv \in \CC, \;\; \inpro{\cv}{\zv} = \sum_{i=1}^n c_i z_i = 0. 
\]
Any generator matrix of $\dual{\CC}$ is called a \emph{parity-check} matrix of $\CC$.

The finite field with $q$ elements is denoted by $\fq$ where $q$ is a power of a prime number $p$.
The \emph{trace operator} of $\fqm$ over $\fq$ 
is the $\fq$-linear map $\Trqm :  \fqm  \longrightarrow  \fq$  defined for any $x$ in $\fqm$ by
\[
\Trqm (x) = x + x^{q}+ \cdots{} +x^{q^{m-1}}.
\]
Let $\mathfrak{B} = \{b_1,\dots{},b_m\}$ be a basis of $\fqm$ over $\fq$. The \emph{dual} basis, or also called the \emph{trace orthogonal} basis 
of $\mathfrak{B}$  is a basis $\mathfrak{B}^* = \{b^*_1,\dots{},b^*_m\}$ of $\fqm$ over $\fq$ such that for any $i$ and $j$ in $\{1\dots{},m\}$
\[
\Trqm( b_i b^*_j) = \delta_{i,j}
\]
where $\delta_{i,i} = 1$ and $\delta_{i,j} = 0$ when $i \neq j$.
Note that there always exits a dual basis and furthermore 
it is  possible to express any $\alpha$ from $\fqm$ as
\begin{equation}
\alpha = \sum_{i=1}^m \Trqm( \alpha b^*_i) b_i.
\end{equation}

Any univariate polynomial $f\in \fqm[X]$ of the form $f_0 + f_1 X^{q} + \cdots{} + f_k X^{q^{d}}$ where $0\leq  d <m$ is a called 
a $q$-\emph{linearised} polynomial and $d$ is its \emph{$q$-degree}.

Any map $h : U \rightarrow V$ is naturally extended to vectors $\xv \in U^n$ by $h(\xv) = ( h(\xv_1),\dots{},h(\xv_n))$.  
This applies in particular to the cases where $h$ is a polynomial  or is the Frobenius (and trace) operator. 
For any subsets $U \subset \F^n$ and $V \subset \F^n$ 
the notation $U + V$ represents the set $\{ \uv + \vv ~ | ~ \uv \in U \text{ and } \vv \in V \}$.
For any subfield $\K \subseteq \F$ and $\xv$ form $\F^n$ the $\K$-vector space generated by  $\xv$ 
is denoted by $\K \xv$.
For any $U \subset \F^n$ and for any $\Pm \in \GL_n(\K)$ the notation $U \Pm$ is used to denote the 
set $\{ \uv \Pm ~ | ~ \uv \in U \}$.
For any subset $V \subseteq \fqm^n$ and any integer $i \ge 0$ we define $V^{q^i}$ as the set of vectors 
$\vv^{q^i}=(v_1^{q^i},\dots{},v_n^{q^i})$ where $\vv$ describes $V$. Note that when $V$ is a vector space then
$V^{q^i}$  is also a  linear subspace of  $\fqm^n$.

\begin{definition}
The \emph{rank weight} of $\xv \in \fqm^n$  denoted by $\norm{q} \xv$
is  the dimension of the $\fq$-vector space generated by $\{x_1,\dots{},x_n\}$, or equivalently
\begin{equation}
\norm{q} \xv = \dim \sum_{i=1}^n \fq  x_i.
\end{equation}
 \end{definition}
Note that for any $\xv \in \F^{n}$ with  $\norm{q}{\xv} = w$ 
there exists  $\Pm$ in $\GL_n(\fq)$  and $\xv^* \in \fqm^{w}$ 
such that
$\xv \Pm =(\xv^* \mid \zz)$
and 
$\norm{q}{\xv^*} = w$.

Finally, an algorithm $D : \F^n \rightarrow \CC$ is said to 
decode $t$ errors in a code $\CC \subset \F^n$ if for any $\cv \in \CC$ and for any $\ev \in \F^n$ 
such that $\norm q \ev \le t$ we have $D(\cv + \ev) = \cv$. 
Generally, we call such a vector $\ev$ an \emph{error} vector.

\section{Gabidulin Codes} \label{sec:gab}

We now introduce an important family of codes known for having an efficient decoding algorithm for the rank metric.

\begin{definition}[Gabidulin code]
Let $\gv$ in $\fqm^{n}$ such that $\norm{q} \gv =n$.
The Gabidulin code $\gab{k}{\gv}$ of length $n$ and dimension $k$ 
is the $\fqm$-linear subspace of $\fqm^n$  defined by 
\begin{equation}
\gab{k}{\gv} = \left \{~ 
f(\gv)   ~|~ f = f_0 + f_1 X^{q} + \cdots{} + f_k X^{q^{k-1}} \in \fqm[X]
~\right \}.
\end{equation}
Equivalently, a generator matrix of $\gab{k}{\gv}$ is given by $\Gm$ where
 \begin{equation} \label{gab:genmat}
\Gm=
\begin{pmatrix}
g_{1} & \cdots{} & g_{n} \\
g_{1}^{q} & \cdots{} & g_{n}^{q} \\
\vdots{}      &              &  \vdots{} \\
g_{1}^{q^{k-1}} & \cdots{} & g_{n}^{q^{k-1}}
\end{pmatrix}.
\end{equation}
\end{definition}
Gabidulin codes are known to possess a fast decoding algorithm that 
can decode errors of weight $t$ provided that $t \leq \lfloor \frac{1}{2}(n-k) \rfloor$.
Furthermore the dual of a Gabidulin code $\gab{k}{\gv}$ is also a Gabidulin code (see for instance \cite{G85,GPT91,B03}).

\begin{proposition}
The dual of  $\gab{k}{\gv}$ is the Gabidulin code $\gab{n-k}{\hv^{q^{-(n-k - 1)}}}$ where $\hv$ belongs to 
 $\dual{\gab{n-1}{\gv}}$ and $\norm{q} \hv = n$.
\end{proposition}
We also have the following proposition.  

\begin{proposition}\label{prop:equivgab}
For any  $\Pm$ in $\GL_n(\fq)$ and for any Gabidulin code $\gab{k}{\gv} \subset \fqm^n$
with $\norm{q} \gv = n$ then 
\begin{equation}
 \gab{k}{\gv} \Pm = \gab{k}{\gv \Pm}.
\end{equation}
\end{proposition}
\begin{proof}
The proof of this proposition comes directly from the fact that for any positive integer $i$, and for any $\Pm$ in $\GL_n(\fq)$,
\[
\left(\gv \Pm \right)^{ q^i }=\gv^{q^i}\Pm
\]
 
\end{proof}

We gather important algebraic properties about Gabidulin codes in order to 
explain why many attacks occur when the underlying code is a Gabidulin code $\gab{k}{\gv}$. 
One key property is 
that Gabidulin codes can be easily distinguished from random linear codes. This singular behaviour has been precisely exploited by Overbeck \cite{O05,O05a,O08} to mount attacks. For that purpose we introduce the operator $\Lambda_i$ defined for any linear vector subspace $U \subseteq \fqm^n$ by 
\begin{equation}
\Lambda_i(U) = U + U^q + \cdots{} + U^{q^i}.
\end{equation} 
This operator can also be defined over matrices in an obvious manner. For instance a generator matrix of $\gab{k}{\gv}$
is $\Lambda_{k-1}(\gv)$. This implies in particular the next proposition.

\begin{proposition} \label{prop:dsg_gab}
For any  $i \ge 0$,
$\Lambda_i\left(\gab{k}{\gv}\right)
= 
\gab{k+i}{\gv}$
which implies in particular that
\[ \dim \Lambda_i \left(\gab{k}{\gv}\right) =  \min \{ k + i, n\}.
\]
\end{proposition}
The importance of Proposition~\ref{prop:dsg_gab} becomes clear when we compare it to the case of random codes.

\begin{proposition}
Let $\CA \subset \fqm^n$ be a code generated by a randomly drawn  matrix from $\MS{k}{n}{\fqm}$ then with a high probability
\begin{equation}
\dim \Lambda _{i}(\CA) = 
\min\big \{(i+1)k, n\big\}
\end{equation}
\end{proposition}

\begin{remark}
Another way of understanding the previous proposition is to observe that if $\CA$ is random code then $\dim \CA \cap \CA^q = 0$
whereas for Gabidulin codes we would obtain 
\[
\dim \gab{k}{\gv} \cap \gab{k}{\gv}^q = k-2.
\]
\end{remark}
Thus there is property  that can be computed  in polynomial time such that it distinguishes  between a Gabidulin code 
and a random code. This important fact has  been used successfully in the cryptanalysis of several 
encryption schemes \cite{COT14,CGGOT14,OT15}.

\section{Faure-Loidreau Encryption Scheme} \label{sec:scheme}

\paragraph{\bf Key generation.} Throughout this step, besides the fields $\fq$ and $\fqm$, another 
field $\L$ is considered where $\L$ is the extension of $\fqm$  
of degree $u > 1$,  and three integers $k$, $n$ and $w$
such that $u < k < n$ and 
\begin{equation}
n - k > w  > \left \lfloor \frac{n-k}{2} \right \rfloor.
\end{equation}

\begin{enumerate}
\item Pick at random $\gv \in \fqm^n$ with $\norm{q}{\gv} = n$ and let 
$\Gm \in \MS{k}{n}{\fqm}$ be the generator matrix of 
$\gab{k}{\gv} \subset \fqm^n$ as in \eqref{gab:genmat}

\item Pick at random $\xv \in \L^k$ such that $\{x_{k-u+1},\dots{}, x_{k} \}$ form a basis of $\L$ over 
$\fqm$  

\item Generate randomly  $\sv \in \L^w$ with $\norm{q}{\sv} = w$ and $\Pm \in \GL_n(\fq)$   and then
compute $\zv \in \L^n$ defined as
\begin{equation}
\zv = \left( \sv \mid \zz \right) \Pm^{-1}.
\end{equation} 
\end{enumerate}
The private key is $\left(\xv, \Pm \right)$ and the public key is $(\gv,k,\Kv,\tp)$ where
\begin{equation}
\Kv = \xv \Gm + \zv  ~~~~\text{ and }~~~~ \tp = \left \lfloor \frac{n-w-k}{2} \right \rfloor.
\end{equation}

\paragraph{\bf Encryption.} 

A plaintext here is a vector $\mv = (m_1,\dots{},m_k)$ belonging to $\fqm^{k}$ such that 
$m_i = 0$ when $i \in  \{k - u + 1,\dots{},k\}$. 
To encrypt then $\mv$ one randomly generates $\alpha \in \L$ and $\ev \in \fqm^{n}$ such that 
$\norm{q}{\ev} \leq \tp$. The ciphertext is the vector $\cv \in \fqm^n$ defined by
\begin{equation}
  \cv=\mv \Gm + \Tr(\alpha \Kv) + \ev.
\end{equation}

\paragraph{\bf Decryption.} 

The receiver computes first $\cv \Pm$ that is to say 
\begin{eqnarray}
\cv \Pm &=& \mv \Gm \Pm + \Tr \big(\alpha \xv \Gm \Pm + \alpha \zv \Pm \big) + \ev \Pm \\
& = & \left( \mv + \Tr(\alpha \xv) \right) \Gm \Pm + \big( \Tr \left(\alpha \sv\right) \mid \zz \big)+ \ev \Pm 
\end{eqnarray}

Let $\Gm^\prime$ be the $k \times (n-w)$ matrix obtained by removing the first $w$ columns of $\Gm\Pm$ 
and let $\ev'$ and $\cv'$ be respectively the restriction of $\ev \Pm$ and $\cv \Pm$  to the
last $n-w$ coordinates.  We then have 
\begin{equation}
\cv'= \left( \mv +  \Tr(\alpha \xv) \right)\Gm^\prime + \ev'.
\end{equation}
Using the fact that $\Gm^\prime$ generates a Gabidulin code of length $n-w$ and dimension $k < n - w$ and 
since $\norm{q}{\ev'} \leq \norm{q}{\ev} \leq \lfloor \frac{1}{2}(n-w-k) \rfloor$, 
it is possible to recover $ \mv' = \mv +  \Tr \left (\alpha \xv\right) $ by applying a decoding algorithm. 
Since by construction $\mv\in \fqm^k$ is chosen so that $m_i = 0$ when $i \in  \{k-u+1,\dots{},k\}$ then
by choosing a dual basis $\{x^*_{k-u+1}, ..., x^*_k\}$ of $\{x_{k-u+1},\dots{}, x_{k} \}$ the value of $\alpha$ can be computed as the following
\[
\sum_{i = k - u + 1}^k m'_i x^*_i = \sum_{i = k - u + 1}^k \Tr(\alpha x_i) x^*_i = \alpha.
\] 
Once $\alpha$ is recovered, the plaintext $\mv$ is then equal to $\mv' - \Tr\left(\alpha \xv\right)$.

\section{Polynomial-Time Key Recovery Attack when $w \le \frac{u}{u+1} (n - k)$} \label{sec:attack}

In this section, we show that it is possible to recover an alternative private key from the public data $\Kv$ and $\Gm$ when the condition $w \le \frac{u}{u+1} (n - k)$ holds. 
We start by remarking that 
if an attacker $\A$ is able to find a matrix $\Tm \in \GL_n(\fq)$ and $\zv^* \in \L^w$ such that 
\[
\zv  \Tm = (\zv^* \mid \zz )  \text{ and }  \norm{q}{\zv^*} = w  
\]
then $\A$ can fully recover $\xv \in \L^k$ by solving 
only the last $n - w$  equations of the following linear system  (see Algorithm~\ref{attack} for more details)
\begin{equation} \label{KT}
\Kv \Tm = \xv \Gm \Tm + (\zv^* \mid \zz ).
\end{equation}

In the sequel, we describe a way to obtain $\xv$ by finding such a matrix $\Tm$.
The first step is to  consider a basis $\gamma_1,\dots{}, \gamma_u$ of $\L$ viewed as a vector space over $\fqm$ of dimension $u > 1$. 
For any $i \in \{1,\dots{}, u \}$ we set $\Kv_i= \Tr(\gamma_i \Kv)$. 
Lastly, let $\Cpub \subset \fqm^n$  be the (public) code  generated by
$\Kv_1,\dots{},\Kv_u$ and $\gab{k}{\gv}$, that is to say
\begin{equation}
\Cpub = \gab{k}{\gv} + \sum_{i=1}^u \fqm \Kv_i.
\end{equation}

\begin{remark}
$\Cpub$ is defined by the generator matrix $\Gp$ where 
\begin{equation} \label{def:Gp}
\Gp = \left( \begin{matrix}
\Gm \\
\Km_1 \\
\vdots{} \\
\Km_u
\end{matrix}
\right)					
\end{equation}
\end{remark}

For all $ i \in  \{ 1,\dots{}, u \}$ let us set  $\vv_i  = \Tr(\gamma_i \zv)$ and 
$\bv_i= \left ( \Tr(\gamma_i \sv) \mid \zz \right) \in \fqm^n$. By construction,
we also have the equality
\begin{equation} \label{eq:vPb}
\vv_i \Pm = \bv_i.
\end{equation}

\begin{lemma}
Let us define $\CB = \sum_{i=1}^m \fqm \bv_i$ then we have 
\[
\Cpub \Pm = \gab{k}{\gv \Pm} + \CB.
\]
\end{lemma}

\begin{proof}
Set $\xv_i = \Tr ( \gamma_i \xv) \in \fqm^k$.
It is sufficient to use Proposition \ref{prop:equivgab} and to observe that 
\begin{eqnarray*}
\Kv_i \Pm &=& \Tr(\gamma_i \xv) \Gm \Pm + \left ( \Tr(\gamma_i \sv) \mid \zz \right) \\
  &=&  \xv_i \Gm \Pm + \bv_i  
 \end{eqnarray*}
 
\end{proof}

\begin{proposition} \label{prop:final}
Let $f =n-w-k-1$ and assume that $\dim \Lambda_f(\CB) = w $. The code 
$\dual{\Lambda_f(\Cpub)}$  is then of dimension $1$ generated by  
$\left( \zz \mid \hv \right) \Pm^T$ where $\hv \in \fqm^{n-w}$ and $\norm{q}{\hv} = n-w$.

Furthermore, for any $\widetilde{\hv} \in \dual{\Lambda_f(\Cpub)}$ with $\widetilde{\hv} \ne \zz $ and 
for any $\Tm \in \GL_{n}(\fq)$ such that 
\begin{equation} \label{alt:key}
\widetilde{\hv}(\Tm^{-1})^T =\left(\ZZ \mid \hv^\prime \right)
\end{equation}
where $\hv^\prime \in \fqm^{n-w}$, there exists  $\zv^{*} \in \fqm^w$ with $\norm{q}{\zv^{*}}=w$ 
such that $\zv \Tm = \left( \zv^{*} \mid \zz \right)$. 
\end{proposition}

\begin{proof}
Let us decompose $\Gm \Pm$  as  $( \Lm ~|~\Rm)$ where 
$\Lm \in \MS{k}{w}{\fqm}$ and $\Rm \in \MS{k}{n-w}{\fqm}$. Let $\Bm \in \MS{u}{w}{\fqm}$ be the matrix where the $i$-th row is composed by the $w$ first components of $\bv_i$. 
Note that $\Gp \Pm$ where $\Gp$ is defined as in \eqref{def:Gp} 
is a generator matrix of  $\Cpub \Pm$, and  the following equality holds
\begin{equation}
\Gp \Pm
=
\begin{pmatrix}
\Lm & \Rm \\
\Bm & \ZZ
\end{pmatrix}.
\end{equation}
Hence  $\Lambda_f(\Gp \Pm) =  \Lambda_f(\Gp) \Pm$ is a generator matrix of the code 
$\Lambda_f(\Cpub \Pm) = \Lambda_f(\Cpub) \Pm$  which satisfies the equality
\[
\Lambda_f(\Gp) \Pm 
=
\begin{pmatrix}
\Lambda_f(\Lm) &  \Lambda_f(\Rm) \\
\Lambda_f(\Bm) & \ZZ
\end{pmatrix}.
\]
The fact that $\Rm$ generates an $(n-w, k)-$Gabidulin code  implies that
 \[
 \rank \left (\Lambda_f(\Rm ) \right )= k + f = n - w - 1.
 \]
 Consequently, there exists $\hv \in \fqm^{n-w}$ with $\norm{q}{\hv} = n-w$ that satisfies 
 $\Lambda_f\left(\Rm\right) \hv^{T} = \zz$. Furthermore, the equality
 $\dim \Lambda_f(\CB) = \Lambda_f(\Bm)$ holds which implies that
\[
\dim{\Lambda_f(\Cpub) \Pm }
= \rank \left (\Lambda_f(\Bm ) \right ) + \rank \left (\Lambda_f(\Rm ) \right )
= k + f + w 
= n - 1. 
\]
This means that $\left( \zz \mid \hv \right) $ generates actually the full space 
$\pdual{\Lambda_f(\Cpub) \Pm}$ 
which is equivalent to say  $\left( \zz \mid \hv \right) \Pm^T $ generates 
$\dual{\Lambda_f(\Cpub)}$.

For the second part of the proposition, let $\widetilde{\hv}$ be any element from $\dual{\Lambda_f(\Cpub)}$ with $\widetilde{\hv} \ne \zz $ and let  $\Tm$ be  in $\GL_{n}(\fq)$ such that \eqref{alt:key} holds
for some $\hv^\prime$ in $\fq^{n-w}$.
There exists an element $\alpha$ in $\fqm$ such that $\widetilde{\hv} = (\zz \mid \alpha \hv)\Pm^T$. 
Consider matrices  
$\Am_1,\Am_2$, $\Am_3$ and $\Am_4$ such that $\Am_1 \in \MS{w}{w}{\fq}$  and $\Am_4 \in \MS{(n-w)}{(n-w)}{\fq}$ so that we have
\[
\Tm^{-1} \Pm =
\left(
\begin{matrix}
\Am_1 & \Am_2 \\
\Am_3 & \Am_4
\end{matrix}
\right).
\]
We have then the  following equalities
\begin{equation}\label{reduction2}
(\zz \mid \hv^{\prime})
=
\widetilde{\hv} (\Tm^{-1})^T 
= 
\left(\zz \mid \alpha \hv \right)\Pm^T (\Tm^{-1})^T 
=
\left(\zz \mid \alpha \hv \right)\left(\Tm^{-1} \Pm \right)^T  
\end{equation}
 It follows from \eqref{reduction2} that $\hv \Am_{2}^T=\ZZ $ and hence $\Am_2 = \ZZ$ since $\norm{q} {\hv} = n-w$. So we can write 
 \[
 \Tm^{-1} \Pm =
\left(
\begin{matrix}
\Am_1 & \ZZ \\
\Am_3 & \Am_4
\end{matrix}
\right).
\]
We deduce that $\Pm^{-1} \Tm =
\left(
\begin{matrix}
\Am^{-1}_1 & \ZZ \\
- \Am^{-1}_4 \Am_3 \Am^{-1}_1  & \Am^{-1}_4
\end{matrix}
\right)= \left(
\begin{matrix}
\Am^\prime & \zz \\
\Cm^\prime & \Dm^\prime
\end{matrix}
\right)
$ and consequently, we get
\[
\zv \Tm = \left( \sv \mid \zz \right) \Pm^{-1} \Tm =\left(\sv \mid \zz \right)\left(
\begin{matrix}
\Am^\prime & \ZZ \\
\Cm^\prime & \Dm^\prime
\end{matrix}
\right)
=\left(\sv \Am^\prime \mid \zz \right).
\]
So by letting $\zv^* = \sv \Am^\prime = \sv \Am_1^{-1}$ we have proved the proposition.
\end{proof}

Proposition \ref{prop:final} shows that an equivalent key can be found in polynomial time by simply 
using a non zero element of $\dual{\Lambda_f(\Cp )}$. We now prove our main result stated in the introduction
which shows the weakness of the system.

\setcounter{theorem}{0}

\begin{theorem}
If the $\fqm$-vector  space generated by $\vv_1,\dots{},\vv_u$ denoted by $V$ satisfies the property 
\[
 \dim \Lambda_{n - w - k - 1}(V) = w
\]
then the private key $(\xv,\zv)$ can be recovered from $(\Gm,\Kv)$ with $O(n^3)$ operations  in the field $\L$
\end{theorem}

\begin{proof}
Firstly, note that from \eqref{eq:vPb} we know that $V \Pm = \CB$.
Algorithm \ref{attack} gives the full description of the attack and 
provides a proof of Theorem~\ref{thm:attack}. Indeed, the attack consists in picking any 
codeword  $\widetilde{\hv}$ from   $\dual{\Lambda_{ n-w-k-1}(\Cpub)}$ and then, by Gaussian elimination, 
we transform $\widetilde{\hv}$ so that there exists $\Tm \in  \GL_{n}(\fq)$ for which we have
\[
\widetilde{\hv}  (\Tm^{-1})^T =\left(\ZZ \mid \hv^\prime \right)
\]
where $\hv^\prime \in \fqm^{n-w}$. From Proposition \ref{prop:final} we know that $\Tm$ is an equivalent key that 
will gives an equality of the form \eqref{KT}, and therefore it is possible by solving a linear system to find $\xv$.
Lastly, the time complexity comes from the fact the operations involved are essentially Gaussian eliminations 
over square matrices with $n$ columns and entries in $\L$.
\end{proof}

An important assumption for the success of the attack is  that  
$\dim \dual{\Lambda_{n-w-k-1}(\Cpub)} = 1$ which was always true in all our experimentations. 
This assumption is true if and only if the equality  $\dim \Lambda_{n-w-k-1}(\CB) = w $ holds,  
which implies to  have $u(n - w - k ) \geq w$, or equivalently
\begin{equation} \label{boundw}
w \le \frac{u}{u+1} (n - k).
\end{equation} 
Assuming that  $\CB$ behaves as a random code then $\dim \Lambda_{n-w-k-1}(\CB) = w $ would hold with high probability as long as \eqref{boundw} is true.
The parameters proposed in \cite{L07} satisfy~\eqref{boundw}. Furthermore, the analysis given in \cite{L07} implies 
to take $u \ge 3$.
We implemented the attack with Magma V2.21-6 and the secret key $\xv$ was found in 
less than $1$ second confirming the efficiency of the approach.

\begin{remark}
Let us observe that taking 
$w > \frac{u}{u+1} (n - k)$  implies for $\tp$ to be very small since we have
\begin{equation} \label{tp:success}
\tp \leq \frac{1}{2} (n-w-k) < \frac{1}{2} \left( \frac{n-k}{u+1} \right).
\end{equation}
For instance, with parameters proposed in \cite{L07} we would have $\tp \leq 3$.
Consequently the values of $n$, $k$ and $m$  have to be changed so that general 
decoding attacks fail \cite{GRS16}. Let us notice that this situation is quite similar to the counter-measures proposed in \cite{RGH10,L10} to resist  to Overbeck's attack.
But  the strength of this reparation deserves a thorough analysis.
\end{remark}

\begin{algorithm}[ht]
    \caption{Key recovery of Faure-Loidreau scheme where the public key is $(\Gm,\Kv)$} 
    \begin{algorithmic}[1]
    	\State{$\{ \gamma_1,\dots{}, \gamma_u \} \gets$  arbitrary basis of $\L$ viewed as a linear space over $\fqm$}
        \ForAll{$1 \leq i \leq u$}
            \State{$K_i \gets \Tr(\gamma_i \Kv)$}
        \EndFor
        \State{Let $\Cpub \subset \fqm^n$ be the code generated by $\Gp$} \Comment{$\Gp$ is defined as in \eqref{def:Gp}}
	\If{$\dim  \dual{\Lambda_{ n-w-k-1}(\Cpub)} = 1$}
		\State{Pick at random $\widetilde{\hv}  \in  \dual{\Lambda_{ n-w-k-1}(\Cpub)}$}
		\State{Compute $\Tm \in  \GL_{n}(\fq)$ and $\hv^\prime \in \fqm^{n-w}$  such that
		\[ \widetilde{\hv}  (\Tm^{-1})^T =\left(\ZZ \mid \hv^\prime \right) \] }
		\State{$\Kv^* \gets \Kv \Tm$}  \Comment{ $\Kv^* = \left ( \Kv^*_1,\dots{},\Kv^*_n \right)\in \L^n$} 	
		\State{$\Gm^* \gets \Gm \Tm$} \Comment{$\Gm^* = ( g^*_{i,j} ) \in \MS{k}{n}{\fqm}$ }
		\State{Solve the linear system where $(X_1,\dots{},X_k)$ are the unknowns
		$$
			(\mathcal{L}) :  \left \{
				\begin{array}{rcl}
					\Kv^*_{w+1} &     =      & g^*_{1,w+1} X_1  + \cdots{} + g^*_{k,w+1} X_k\\
					                        &\vdots{} &  \\
					\Kv^*_{n} &     =         & g^*_{1,n}  X_1 + \cdots{} + g^*_{k,n} X_k 
				\end{array}
			\right.
		$$}
		\State{$\zv \gets \Kv -  \xv \Gm$ where $\xv$ is the \emph{unique} solution of  $(\mathcal{L})$}  
	\EndIf
	\State{\Return $(\xv,\zv)$}
    \end{algorithmic}
    \label{attack}
\end{algorithm}

\begin{table}[h]
\begin{center}
\caption{Bound on $w$ with parameters taken from \cite{L07} ($m = n$).} \label{tab:param}
\begin{tabular}{@{}rrrrr@{}} \toprule 
$n$    & $k$    & $u$  & $w$ &  $\frac{u}{u+1} (n - k)$\\  \midrule
$56$  & $28$  &  $3$ &   $16$ & $21$  \\      
$54$  & $32$  & $4$  &   $13$  & $17$ \\   
 \bottomrule
\end{tabular}

\end{center}
\end{table}

\section{Conclusion}

Faure and Loidreau proposed a rank-metric encryption scheme based on Gabidulin codes related to the problem of the linearized polynomial reconstruction.
We showed that the scheme is vulnerable to a polynomial-time key recovery attack by using Overbeck's techniques applied on an appropriate public code.

Our attack assumes that parameters are chosen so that 
$w \leq \frac{u}{u+1} \left( n - k \right)$ 
which was always the case in \cite{FL05,L07}. We have also seen that taking $w > \frac{u}{u+1} \left( n - k \right)$  implies to choose $\tp < \frac{1}{2} \left( \frac{n-k}{u+1} \right)$ which exposes further the system to general decoding attacks like \cite{GRS16}. Hence it imposes to increase the key sizes and consequently reduces the practicability of the scheme while offering no assurance that the scheme is still secure.
The best choice from a designer's point of view would be to take  $u$ as small as possible
but a thorough analysis has to be undertaken in light of the connections with the reparations proposed in \cite{RGH10,L10}.  This point is left as an open question in our paper and breaking this kind of parameters would lead arguably to a cryptanalysis of \cite{RGH10,L10}, and to an algorithm that decodes  Gabidulin codes  beyond the bound $\frac{u}{u+1} \left( n - k \right)$.

\section{Acknowledgements}

The authors would like to thank Pierre Loidreau for helpful discussions and for bringing reference \cite{LO06} to our attention.

\bibliographystyle{alpha}
\bibliography{codecrypto}
\end{document}